\PassOptionsToPackage{hyphens}{url}
\documentclass[a4paper,USenglish,cleveref,autoref,numberwithinsect]{lipics-v2021} 

\usepackage{csquotes}
\usepackage{siunitx}
\usepackage{microtype}
\usepackage[shortcuts]{extdash}

\hideLIPIcs
\nolinenumbers
\title{Holes in Convex and Simple Drawings} 

\author{Helena Bergold}{Institute of Computer Science, Freie Universit{\"a}t Berlin, Germany}{helena.bergold@fu-berlin.de}{https://orcid.org/0000-0002-9622-8936}{Supported by DFG-Research Training Group 'Facets of Complexity' (DFG-GRK~2434).}

\author{Joachim Orthaber}{Institute of Algorithms and Theory, Graz University of Technology, Austria}{orthaber@tugraz.at}{https://orcid.org/0000-0002-9982-0070}{Supported by the Austrian Science Fund (FWF) grant W1230.}

\author{Manfred Scheucher}{Institute of Mathematics, Technische Universit{\"a}t Berlin, Germany}{scheucher@math.tu-berlin.de}{https://orcid.org/0000-0002-1657-9796}{Supported by DFG Grant SCHE~2214/1-1.}

\author{Felix Schr{\"o}der}{Institute of Mathematics, Technische Universit{\"a}t Berlin, Germany\\ Department of Applied Mathematics, Charles University, Czech Republic}{schroder@kam.mff.cuni.cz}{https://orcid.org/0000-0001-8563-3517}{Supported by the GA\v{C}R Grant no. 23-04949X.}

\authorrunning{H.~Bergold, J.~Orthaber, M.~Scheucher, and F.~Schr{\"o}der}

\ccsdesc[500]{Mathematics of computing~Discrete mathematics}
\ccsdesc[500]{Theory of computation~Computational geometry}

\keywords{simple topological graph, convexity hierarchy, \texorpdfstring{$k$}{k}-gon, \texorpdfstring{$k$}{k}-hole, empty \texorpdfstring{$k$}{k}-cycle, Erd\H{o}s--Szekeres theorem, Empty Hexagon theorem, planar point set, pseudolinear drawing}

\relatedversion{A short version of this paper appeared in the Proceedings of the 32nd International Symposium on Graph Drawing and Network Visualization (GD 2024), see~\cite{gd2024}.}

\usepackage[font=small,labelfont=bf]{caption}
\usepackage[font=small,labelfont=normalfont,labelformat=simple]{subcaption}

\newcommand{\calC}{\mathcal{C}}
\newcommand{\calD}{\mathcal{D}}
\newcommand{\calT}{\mathcal{T}}

\begin{document}

\maketitle

\begin{abstract}
Gons and holes in point sets have been extensively studied in the literature.
For simple drawings of the complete graph a generalization of the Erd\H{o}s--Szekeres theorem is known and empty triangles have been investigated.
We introduce a notion of $k$-holes for simple drawings and survey generalizations thereof, like empty $k$\=/cycles.
We present a family of simple drawings without $4$-holes and prove a generalization of Gerken's empty hexagon theorem for convex drawings.
A crucial intermediate step is the structural investigation of pseudolinear subdrawings in convex drawings.
With respect to empty $k$-cycles, we show the existence of empty $4$-cycles in every simple drawing of $K_n$ and give a construction that admits only $\Theta(n^2)$ of them.
\end{abstract}

\section{Introduction}\label{sec:Introduction}

A classic theorem from combinatorial geometry is the Erd\H{o}s--Szekeres theorem~\cite{ErdosSzekeres1935}.
It states that for every $k \in \mathbb{N}$ every sufficiently large point set in general position (that is, no three points on a line) contains a subset of $k$ points that are the vertices of a convex polygon, a so called \emph{$k$-gon}.
In this article we focus on a prominent variant of the Erd\H{o}s--Szekeres theorem suggested by Erd\H{o}s himself~\cite{erdos75_problems}, which asks for the existence of \emph{empty} $k$-gons, also known as $k$-holes.
A \emph{$k$-hole} $H$ in a point set $P$ is a $k$-gon with the property that there are no points of~$P$ in the interior of the convex hull of~$H$.
It is known that every sufficiently large point set contains a $6$-hole~\cite{Gerken2008,Nicolas2007} and that there are arbitrarily large point sets without $7$-holes~\cite{Horton1983}.

Point sets in general position are in correspondence with \emph{geometric drawings} of the complete graph where vertices are mapped to points and edges are drawn as straight-line segments between the vertices.
In this article we generalize the notion of holes to simple drawings of the complete graph~$K_n$.
In a \emph{simple drawing}, vertices are mapped to distinct points in the plane (or on the sphere) and edges are mapped to simple curves connecting the two corresponding vertices such that two edges have at most one point in common, which is either a common vertex or a proper crossing.
In the course of this article, we will see that many important properties do not depend on the full drawing but only on the underlying combinatorics, more specifically, on the isomorphism class of a drawing.
We call two simple drawings of the same graph \emph{isomorphic}\footnote{This isomorphism is often referred to as \enquote{weak isomorphism} since there also exist stronger notions.}
if there is a bijection between their vertex sets such that the corresponding pairs of edges cross.
Note that this isomorphism is independent of the choice of the outer cell and thus only encodes the simple drawing on the sphere.

To study $k$-holes, we first extend the notion of $k$-gons to simple drawings of~$K_n$.
A~\emph{$k$\=/gon}\footnote{%
We keep the terminology from the geometric setting and trust that this does not lead to any confusion.}%
~$\calC_k$ is a subdrawing isomorphic to the geometric drawing on $k$ points in convex position; see \Cref{fig:convex} for a depiction of an $n$-gon.
In terms of crossings, a $k$-gon $\calC_k$ is a (sub)drawing with vertices $v_1, \ldots, v_k$  such that $\{v_i,v_\ell\}$ crosses $ \{v_j,v_m\}$ if and only if $i<j<\ell<m$.
In contrast to the geometric setting where every sufficiently large geometric drawing contains a $k$-gon, simple drawings of complete graphs do not necessarily contain $k$-gons~\cite{Harborth98_emptytriangles}.
For example, the twisted drawing $\calT_n$ depicted in \Cref{fig:twisted} does not contain any 5-gon.
In terms of crossings, $\calT_n$~can be characterized as a drawing of $K_n$ with vertices $v_1, \ldots, v_n$ such that $\{v_i, v_m\}$ crosses $\{v_j, v_{\ell}\}$ exactly if $i<j<\ell<m$.
A~theorem by Pach, Solymosi and T\'oth~\cite{PachSolymosiToth2003} states that, for every~$k$, every sufficiently large simple drawing of~$K_n$ contains $\calC_k$ or~$\calT_k$.
The currently best known estimate is due to Suk and Zeng~\cite{sz-2024-upcstg} who showed that every simple drawing of $K_n$ with $n > 2^{9 \cdot \log_2(a)\log_2(b) a^2  b^2}$ contains $\calC_a$ or~$\calT_b$.
Convex drawings, which we define in the next paragraph, are a class of drawings nested between geometric drawings and simple drawings. In particular, convex drawings do not contain $\calT_5$ as a subdrawing.
Hence every convex drawing of $K_n$ contains a $k$-gon $\calC_k$ for some $k = (\log n)^{1/2-o(1)}$.

\begin{figure}[htb]
\centering
\subcaptionbox{\label{fig:convex}}[.49\textwidth]{\includegraphics[page=1]{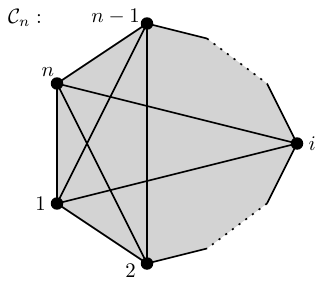}}
\subcaptionbox{\label{fig:twisted}}[.49\textwidth]{\includegraphics[page=1]{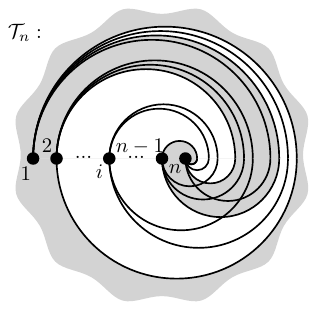}}
\caption{A drawing of \subref{fig:convex} an $n$-gon $\calC_n$ and \subref{fig:twisted} a twisted $\calT_n$ for $n \geq 4$. The largest hole in these drawings, an $n$-hole and a $4$-hole respectively, is shaded in gray.}
\label{fig:convextwisted}
\end{figure}

In the last decades, holes were intensively studied for the setting of point sets.
Our focus lies on determining the existence of holes in convex drawings, the most general class of the convexity hierarchy introduced by Arroyo, McQuillan, Richter, and Salazar~\cite{ArroyoMRS2022}, which gives a more fine-grained layering between geometric drawings and simple drawings.
The basis to define convexity are \emph{triangles}, which are subdrawings induced by three vertices.
Since in a simple drawing incident edges do not cross, a triangle separates the plane (respectively the sphere) into two connected components.
The closure of each of the components is called a \emph{side} of the triangle.
A side $S$ is \emph{convex} if, for every pair of vertices in~$S$, the connecting edge is fully contained in~$S$.
A simple drawing $\calD$ of $K_n$ is
\begin{itemize}
\item \emph{convex} if every triangle in $\calD$ has a convex side;
\item \emph{h-convex} (hereditarily convex) if there is a choice of a convex side $S_T$ for every triangle $T$ such that, for every triangle $T'$ contained in $S_T$, it holds that $S_{T'} \subseteq S_T$;
\item \emph{f-convex} (face convex) if there is a marking face~$F$ in the plane such that for all triangles the side not containing $F$ is convex. 
\end{itemize}

The class of f-convex drawings is related to pseudolinear drawings.
A \emph{pseudolinear drawing} is a simple drawing in the plane such that the edges can be extended to an arrangement of pseudolines.
A \emph{pseudoline} is a simple curve partitioning the plane into two unbounded components and in an \emph{arrangement} each pair of pseudolines has exactly one point in common, which is a proper crossing.
As shown by Arroyo, McQuillan, Richter, and Salazar~\cite{ArroyoMRS2017_pseudolines}, a simple drawing of $K_n$ is pseudolinear if and only if it is f-convex and the marking face $F$ is the unbounded face.
For more details on the convexity hierarchy and the classes it contains, we refer the reader to~\cite{ArroyoMRS2017_pseudolines,ArroyoMRS2022,ArroyoRS21_pscircle,BFSSS_TDCTCG_2022}.

Before we define $k$-holes, consider the case of 3-holes, also known as empty triangles.
A~triangle is \emph{empty} if one of its two sides does not contain any vertex in its interior.
Harborth~\cite{Harborth98_emptytriangles} proved that every simple drawing of $K_n$ contains at least two empty triangles and conjectured that the minimum among all simple drawings of $K_n$ is~$2n-4$.
While $2n-4$ is obtained by~$\calT_n$ and all generalized twisted drawings~\cite{GarciaTejelVogtenhuberWeinberger2022}, the best known lower bound is~$n$~\cite{AichholzerHPRSV2015}.

In the geometric setting, the number of empty triangles behaves differently: every point set has $\Omega(n^2)$ empty triangles, and this bound is asymptotically optimal~\cite{BaranyFueredi1987}.
Note that the notion of empty triangles in point sets slightly differs from the one in simple drawings, where the complement of the convex hull of a point set can be an empty triangle as well.
The class of convex drawings behaves similarly to the geometric setting: the minimum number of empty triangles is asymptotically quadratic~\cite[Theorem 5]{ArroyoMRS2017_pseudolines}.

\subparagraph{Holes in Simple Drawings.}
In the drawing $\calC_k$ with $k \geq 4$, every triangle has exactly one empty side, which is also its unique convex side. 
The \emph{convex side} of $\calC_k$ is the union of convex sides of its triangles; see the gray shaded regions in \Cref{fig:convextwisted}.
Given a $k$-gon $\calC_k$ in a simple drawing of~$K_n$, we call vertices in the interior of the convex side of $\calC_k$ \emph{interior vertices}.
A~\emph{$k$-hole} in a simple drawing of~$K_n$ is a $k$-gon that has no interior vertices.
For example, the vertices $1$, $2$, $n-1$, and $n$ in $\calT_n$ form a $4$-hole; marked gray in \Cref{fig:twisted}.
In convex drawings, as in the geometric setting, edges from an interior vertex to a vertex of $\calC_k$ and edges between two interior vertices are contained in the convex side of~$\calC_k$ \cite[Lemma 3.5]{ArroyoMRS2022}; see also \Cref{sec:convexholes}.

\medskip

In this paper, using the notion of $k$-holes in simple drawings defined above, we resolve the questions of existence of $4$-, $5$- and $6$-holes in simple and convex drawings of~$K_n$.
In~particular, we show the existence of $6$-holes in sufficiently large convex drawings (\Cref{thm:6holes}), generalizing Gerken's empty hexagon theorem~\cite{Gerken2008}.
The key ingredient of the proof is that in a convex drawing every subdrawing induced by a minimal $k$-gon together with its interior vertices is f-convex (\Cref{lem:minimalkgonpseudolinear}).
This allows to transfer various existential results from the geometric, pseudolinear, and f-convex settings to convex drawings.
Besides the existence of 6-holes, we also show the existence of monochromatic generalized 4-holes in two-colored convex drawings (\Cref{cor:genmonochromatic4hole}), generalizing a result on bichromatic point sets~\cite{AHHHFV2010}.
For this we discuss two variants of generalized holes (\Cref{sec:generalized_holes}) in the setting of simple drawings of~$K_n$ and show the existence of empty $4$-cycles, that is, plane cycles of length $4$ such that one side does not contain any interior vertices (\Cref{thm:empty4cycle}).
Furthermore, we construct a simple drawing of $K_n$ that does not contain any two interior-disjoint empty triangles sharing an edge (\Cref{prop:nogen4hole}) and another one containing only $\Theta(n^2)$ empty $4$-cycles (\Cref{prop:empty4-cycle}).

\section{Holes in Convex Drawings}\label{sec:convexholes}

In this section, we show that convex drawings behave similarly to geometric point sets when it comes to the existence of holes.
We show that every sufficiently large convex drawing contains a $6$-hole and hence a $5$-hole and a $4$-hole.
This is tight, as the construction by Horton~\cite{Horton1983} gives arbitrarily large point sets, that is, geometric drawings without $7$-holes.

\begin{theorem}[Empty hexagon theorem for convex drawings]
\label{thm:6holes}
For every sufficiently large $n$, every convex drawing of $K_n$ contains a 6-hole.
\end{theorem}

For the proof we use the existence of $k$-gons in sufficiently large convex drawings~\cite{PachSolymosiToth2003,sz-2024-upcstg}.
Our key lemma is that the subdrawing induced by a minimal $k$-gon together with its interior vertices is f-convex, a fact that had been known only for h-convex drawings~\cite[Lemma 4.7]{ArroyoMRS2022}.
For $k$ fixed, a $k$-gon is \emph{minimal} if its convex side does not contain the convex side of another $k$-gon.

Arroyo, McQuillan, Richter and Salazar \cite[Section 3]{ArroyoMRS2022} started the investigations of interior vertices of $k$-gons.
An important part is their Lemma~3.5, which we use in the following.

\begin{lemma}[cf.\ {\cite[Lemma~3.5]{ArroyoMRS2022}}]
\label{lem:arroyo_verticesingon}
Let $\calC_k$ be a $k$-gon in a convex drawing of $K_n$ with vertices $v_1, \ldots, v_k$ and $k \geq 4$.
Then for every two vertices $u,v$ contained in the convex side of~$\calC_k$ the edge $\{u,v\}$ is contained in the convex side of~$\calC_k$.
\end{lemma}

Note that in a $k$-gon $\calC_k$ the edges on its convex hull form a \emph{plane $k$-cycle}, that is, a cycle of length $k$ that does not cross itself.
This plane cycle divides the plane into two connected components whose closures we call \emph{sides}.
Furthermore, all \emph{chords} of that cycle, that is, edges between non-adjacent vertices of the cycle lie on the same side of the cycle.
On the other hand, if all chords of a plane $k$-cycle lie on the same side of it, then they cross each other in the exact same pattern as in a $k$-gon~$\calC_k$.

\begin{observation}\label{obs:k-gon-is-plane-cycle-with-chords}
A $k$-gon $\calC_k$ is equivalent to a plane $k$-cycle that has all chords on the same side, which is the convex side of $\calC_k$.
\end{observation}

For the sake of readability, we refer to the vertices $v_1, \ldots, v_k$ of a $k$-gon with indices modulo~$k$.

\begin{lemma}
\label{lem:vertexinboundarytriangle}
Let $\calC_k$ be a minimal $k$-gon in a convex drawing $\mathcal{D}$ of $K_n$ with vertices $v_1, \ldots, v_k$ and $k \geq 3$.
Then for all~$i$ there are no interior vertices in the convex side of the triangle $\{v_i, v_{i+1}, v_{i+2}\}$.
In particular, every minimal $4$-gon is a $4$-hole and every minimal $3$-gon is an empty triangle.
\end{lemma}

\begin{proof}
Assume there is an interior vertex $v$ in the convex side of the triangle determined by $\{v_i, v_{i+1}, v_{i+2}\}$.
If $k=3$ then, by minimality of $\calC_k = \{ v_1, v_2, v_3 \}$, the side $S_N$ of the triangle $\{ v_1, v_2, v \}$ contained in the convex side of $\calC_k$ cannot be convex.
Hence, there exists a vertex $z$ in the interior of~$S_N$ such that the subdrawing induced by $\{v_1, v_2, v, z\}$ has a crossing~\cite[Corollary~2.5]{ArroyoMRS2022}.
The edge $\{z,v\}$ cannot cross $\{v_1,v_2\}$ since that would contradict the convex side of~$\calC_k$.
Hence, without loss of generality, let the edge $\{z,v_1\}$ cross $\{v,v_2\}$; \Cref{fig:minimal-3-gon} gives an illustration.
Then, however, the edge $\{z,v_1\}$ shows that the triangles $\{ v_2, v_3, v \}$ and $\{ v_1, v, v_3 \}$ both have a unique convex side, which is the side contained in the convex side of $\calC_k$.
This is a contradiction to the minimality of $\calC_k$.
Thus, a minimal $3$-gon is an empty triangle.

For $k \geq 4$, clearly the vertices $v_1, \ldots, v_{i}, \allowbreak{} v_{i+2}, \ldots, v_k$ span a $(k-1)$-gon and the triangle $v_i, v, v_{i+2}$ is not contained in the convex side of that $(k-1)$-gon.
Moreover all chords of $\calC_k$ not involving $v_{i+1}$ lie in the convex side of that $(k-1)$-gon.
It remains to consider edges incident to~$v$.
Let $j \in [k] \setminus \{i,i+1,i+2\}$ be arbitrary but fixed.
By \Cref{lem:arroyo_verticesingon}, the edge $\{v,v_j\}$ does not leave the convex side of $\calC_k$ and, since $\mathcal{D}$ is a simple drawing, $\{v,v_j\}$ crosses $\{v_i,v_{i+2}\}$ and therefore lies in the convex side of the $4$-gon $v, v_i, v_j, v_{i+2}$.
\Cref{fig:smallerkgon} gives an illustration.
This shows that $v_1, \ldots, v_{i}, \allowbreak{} v, \allowbreak{} v_{i+2}, \ldots, v_k$ span a plane $k$-cycle with all chords on the same side and hence, by \Cref{obs:k-gon-is-plane-cycle-with-chords}, they span a $k$-gon~$\calC'_k$.
Furthermore, the convex side of $\calC'_k$ is contained in the convex side of $\calC_k$, implying that $\calC_k$ was not minimal; a contradiction.
\end{proof}

\begin{figure}[htb]
\centering
\subcaptionbox{\label{fig:minimal-3-gon}}[.49\textwidth]{\includegraphics[page=1]{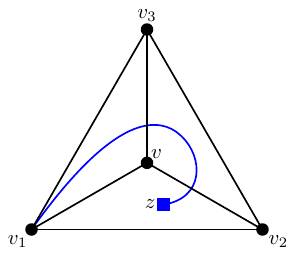}}
\subcaptionbox{\label{fig:smallerkgon}}[.49\textwidth]{\includegraphics[page=2]{figs/smallergons.pdf}}
\caption{\subref{fig:minimal-3-gon} If the convex side of a triangle is not empty, it contains the convex side of another triangle. \subref{fig:smallerkgon} A $k$-gon with an interior vertex $v$ in the convex side of the triangle $v_i,v_{i+1},v_{i+2}$.}
\label{fig:minimal-k-gon}
\end{figure}

\begin{lemma}[Key lemma]
\label{lem:minimalkgonpseudolinear}
Let $\calC_k$ be a minimal $k$-gon in a convex drawing $\calD$ of $K_n$ with \mbox{$n \geq k \geq 3$}.
Then the subdrawing $\calD'$ induced by the vertices in the convex side of~$\calC_k$ is f-convex.
\end{lemma}

\begin{proof}
For $k \leq 4$, by \Cref{lem:vertexinboundarytriangle}, a minimal $k$-gon is empty and thus $\calD'$ is clearly f-convex.

So let $k \geq 5$, let $v_1, \ldots, v_k$ be the vertices of the minimal $k$-gon $\calC_k$ in $\calD$, and let $F$ be a face contained in the non-convex side of~$\calC_k$.
We show that for every triangle spanned by three vertices of the convex side of $\calC_k$, the side not containing $F$ is convex and hence $\calD'$ is f-convex.
Suppose towards a contradiction that there exists a triangle $T$ spanned by vertices $t_1,t_2,t_3$ from the convex side of~$\calC_k$, such that the side not containing $F$ is not convex.
Then this non-convex side $S_N$ of $T$ is the side contained in the convex side of~$\calC_k$.
Since $\calD$ is convex, the other side of $T$, containing $F$ and all vertices $v_1, \ldots, v_k$, is convex and is denoted by~$S_C$.
If we additionally assume that $S_N$ is not contained in (the closure of) a single cell of the subdrawing induced by~$\calC_k$, then some edge $\{v_i,v_j\}$ has a crossing with one of the edges $\{t_\ell,t_m\}$ of $T$.
This shows that $S_C$ is not convex; a contradiction.
Hence, $S_N$ lies in (the closure of) a cell of~$\calC_k$.

Since $\calC_k$ is minimal, by \Cref{lem:vertexinboundarytriangle}, there are no interior vertices in the convex side of a triangle $\{v_i,v_{i+1},v_{i+2}\}$.

Since all cells in the convex side of $\calC_k$ incident to the vertex $v_{i+1}$ are inside this triangle, the vertex $v_{i+1}$ is not part of the triangle $T$ spanned by $t_1, t_2, t_3$.
This holds for every~$i=1,\ldots,k$ and hence the vertices $t_1, t_2, t_3$ are interior vertices of $\calC_k$ and $S_N$ lies in a cell of the convex side of $\calC_k$ that is not covered by the convex side of any triangle $\{v_i, v_{i+1}, v_{i+2} \}$.
Since $S_N$ is not convex, there exists a vertex $z$ in the interior of~$S_N$ such that the subdrawing induced by $\{t_1, t_2, t_3,z\}$ has a crossing~\cite[Corollary 2.5]{ArroyoMRS2022}.
We assume without loss of generality that the edge $\{t_1,z\}$ crosses $\{t_2,t_3\}$.
Moreover, exactly one of the following two conditions holds:
Either the triangle $\{t_1, t_3,z\}$ separates $t_2$ and~$F$ or the triangle $\{t_1, t_2,z\}$ separates $t_3$ and~$F$.
We assume that the former holds as otherwise we exchange the roles of $t_2$ and~$t_3$.
\Cref{fig:withoutboundary} gives an illustration.
	
\begin{figure}[htb]
\centering
\includegraphics[page=1]{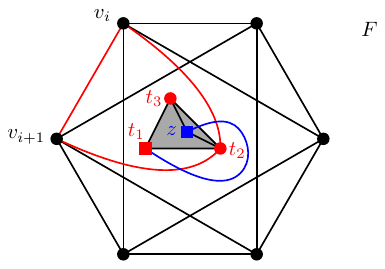}
\caption{The non-convex side (shaded gray) of the triangle $\{t_1, t_2, t_3\}$ (red vertices) is witnessed by the edge $\{t_1, z\}$ (blue), and the triangle $\{t_1, t_3, z\}$ separates $t_2$ and~$F$. Then the triangle $\{t_2, v_i, v_{i+1}\}$ (red edges) has no convex side.}
\label{fig:withoutboundary}
\end{figure}

Now we consider all edges from $t_2$ to the vertices $v_1,\ldots,v_k$ of $\calC_k$.
Since $S_C$ is convex and contains $v_1,\ldots, v_k$, the edges $\{t_2,v_i\}$ are contained in $S_C$.
This shows that none of the edges $\{t_2,v_i\}$ crosses any of the triangle edges and, in particular, they do not cross $\{t_1,t_3\}$.
	
The edges $\{t_2,v_1\},\ldots,\{t_2,v_k\}$ partition the convex side of $\calC_k$ into triangles $\{t_2,v_i,v_{i+1}\}$.
Hence there is an index~$i$ such that the three vertices $t_1,t_3,z$ lie in the side of the triangle $\{t_2,v_i,v_{i+1}\}$, which is contained in the convex side of~$\calC_k$.
However, the edge $\{t_1,z\}$ is not fully contained in this side; a contradiction to its convexity.
Moreover, the other side of that triangle is not convex either:
Since $t_2$ is not inside the triangle $\{v_i,v_{i+1},v_{i+2}\}$ the edge $\{v_i,v_{i+2}\}$ crosses $\{t_2,v_{i+1}\}$, so it is not fully contained in this side.
This is a contradiction to the convexity of the drawing and thus completes the proof.
\end{proof}

Recently, Heule and Scheucher~\cite{HeuleScheucher2024} used SAT to show that every set of 30 points has a 6-hole.
Since their result is about the more general case of pseudoconfigurations of points, it holds for pseudolinear drawings.
To prove \Cref{thm:6holes}, we combine this fact with \Cref{lem:minimalkgonpseudolinear}.

\begin{proof}[Proof of \Cref{thm:6holes}]
Let $\calD$ be a convex drawing of $K_n$ with $ n > 2^{9\cdot 5^2 \log_2 (5) \cdot 30^2 \log_2 (30)}$.
Since convex drawings do not contain the twisted drawing $\calT_5$, it follows from \cite{sz-2024-upcstg} that $\calD$ contains a $30$-gon.
To find a $6$-hole in~$\calD$, we choose a minimal $30$-gon~$\calC_{30}$.
By \Cref{lem:minimalkgonpseudolinear}, the subdrawing $\calD'$ induced by $\calC_{30}$ and its interior vertices is f-convex.
Since the existence of holes is invariant under the choice of the outer cell, we can assume without loss of generality that $\calD'$ is pseudolinear as we may otherwise choose the face $F$ as the unbounded face.
According to \cite{BalkoFulekKyncl2015}, $\calD'$ corresponds to a pseudoconfiguration of points, and hence there exists a $6$-hole~$\calC_{6}$ in~$\calD'$~\cite{HeuleScheucher2024}.
Hence the convex side of $\calC_{6}$ does not contain any vertex of~$\calD'$.
Moreover, every vertex of $\calD$ in the convex side of~$\calC_{6}$ would be an interior vertex of $\calC_{30}$ and therefore belong to~$\calD'$.
This shows that $\calC_{6}$ is also a $6$-hole in $\calD$.
\end{proof}

The existence of $6$-holes further implies the existence of $4$- and $5$-holes.
However, it remains a challenging task to determine the smallest integer~$n(k)$ such that every convex drawing of~$K_n$ with $n \geq n(k)$ contains a $k$-hole for $k \in \{ 5, 6 \}$. The case $k=4$ we resolve below.

For $6$-holes, one can slightly improve the estimate from \Cref{thm:6holes} by utilizing the fact that every $9$-gon in a point set yields a $6$-hole~\cite{Gerken2008}.
As shown in~\cite{Scheucher2022}, this result transfers to pseudolinear drawings.
It then follows from \cite{sz-2024-upcstg} and \Cref{lem:minimalkgonpseudolinear} that every convex drawing of $K_n$ with $n > 5^{\num{18225} \cdot \log_2 (9)}$ contains a $6$-hole.

A similar improvement is possible for $5$-holes: as the textbook proof for the existence of $5$-holes in every $6$-gon of a point set (see for example Section~3.2 in \cite{Matousek2002_book}) applies to pseudolinear drawings, every convex drawing with more than $ 5^{\num{8100} \cdot \log_2 (6)}$ vertices contains a $5$-hole.

For $4$-holes, we can combine the proof of B\'ar\'any and F\"uredi~\cite[Theorem~3.3]{BaranyFueredi1987} for the quadratic number of $4$-holes in point sets and the proof of Arroyo, McQuillan, Richter, and Salazar~\cite[Theorem~5]{ArroyoMRS2017_pseudolines} for the quadratic number of empty triangles in convex drawings to obtain:

\begin{lemma}
\label{lemma:4holes}
Every crossed edge in a convex drawing of $K_n$ is a chord of a $4$-hole, that is, it is one of the crossing edges of a $4$-hole.  
\end{lemma}

\begin{proof}
Let $\calD$ be a convex drawing of~$K_n$.
Let $e$ be an edge that is crossed by another edge~$f$.
The subdrawing induced by the four end-vertices of $e$ and $f$ is a $4$-gon, and we denote it by~$\calC_4$.
We assume that the vertices are labeled with $v_1, v_2, v_3 , v_4$ such that $e = \{v_1,v_3\}$ and $f = \{v_2,v_4\}$.
If $\calC_4$ is minimal, it is a 4-hole by \Cref{lem:vertexinboundarytriangle}.

Hence, we assume that there is an interior vertex $x$ of~$\calC_4$ as illustrated in \Cref{fig:4holes}.
By the properties of a $4$-gon, $x$ lies in the convex side of exactly two of its triangles.
Without loss of generality, we assume that $x$ is in the convex side of the two triangles $\{v_1,v_2,v_3\}$ and $\{v_2,v_3,v_4\}$.
By \Cref{lem:arroyo_verticesingon}, the edges $\{x, v_i\}$ are fully contained in the convex side of~$\calC_4$.
Since the edge $\{x,v_4\}$ is fully contained in the convex side of $\{v_2,v_3,v_4\}$, but has to leave the triangle $\{v_1,v_2,v_3\}$ to get to~$v_4$, it crosses the edge $e = \{v_1,v_3\}$.
Hence $v_1,x,v_3,v_4$ span another $4$-gon in which $\{v_1, v_3\}$ is one of the crossing edges.
Furthermore, since the edges $\{x,v_1\}$, $\{x,v_2\}$, and $\{x,v_3\}$ are fully contained in the convex side of~$\calC_4$, the convex side of the $4$-gon $\{v_1,x,v_3,v_4\}$ is fully contained in the convex side of~$\calC_4$.

This shows that every $4$-gon that is minimal subject to the restriction that $e$ is one of its chords is actually a minimal $4$-gon without restriction.
Consequently, by \Cref{lem:vertexinboundarytriangle}, every crossed edge $e$ gives a $4$-hole whose diagonal is $e$, which completes the proof.
\end{proof}

\begin{figure}[htb]
\centering
\includegraphics{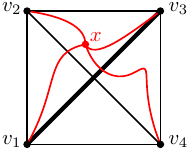}
\caption{If a $4$-gon $\{v_1,v_2,v_3,v_4\}$ with chord $e = \{v_1,v_3\}$ is not empty, then it contains a smaller $4$-gon $\{v_1,x,v_3,v_4\}$ still with chord~$e$.}
\label{fig:4holes}
\end{figure}

Note that there are $\binom{n}{2}$ edges in a drawing of the complete graph, at most $2n-2$ of which are uncrossed \cite{Ringel1964}.
Since every $4$-hole is counted at most twice, the total number of $4$-holes in a convex drawing of $K_n$ is at least $\frac{1}{2} \left( \binom{n}{2} - 2n+2 \right) = \frac{1}{4}n^2 - \frac{5}{4}n + 1$.

Since every drawing of $K_5$ contains a crossing, \Cref{lemma:4holes} also implies that every convex drawing of $K_n$ with $n \geq 5$ contains a $4$-hole.
In contrast to the convex setting, $4$-holes can be avoided in simple drawings as we show in the next section.

\section{Generalized Holes}\label{sec:generalized_holes}

Devillers, Hurtado, K\'arolyi, and Seara~\cite{DevillersHKS2003} showed that sufficiently large two-colored point sets in general position contain a monochromatic $3$-hole and constructed arbitrarily large two-colored sets without monochromatic $5$-holes.
The existence of monochromatic $4$-holes, however, remains a longstanding open problem~\cite[Problem~8.2.7]{BrassMoserPach2005}.
A weaker version was shown by Aichholzer, Hackl, Huemer, Hurtado, and Vogtenhuber~\cite{AHHHFV2010}.
They proved that every two-colored point set $P= A \ \dot\cup \ B$ contains a monochromatic generalized $4$-hole.
A~\emph{generalized \mbox{$k$-}hole} is a simple polygon (not necessarily convex) which is spanned by $k$ points of $P$ and does not contain any point of $P$ in its interior.

To define generalized $k$-holes in simple drawings we consider plane cycles.
Recall that a plane cycle divides the plane into two connected components whose closures we call sides.
An \emph{empty $k$-cycle} in a simple drawing is a plane cycle of length $k$ such that one of its sides contains no vertices in its interior.
For $k=3$ this definition coincides with empty triangles.
Since polygons in point sets can be triangulated, we say that an empty $k$-cycle is an \emph{empty $k$-triangulation} if its empty side is the disjoint union of empty triangles.

Since the proof in \cite{AHHHFV2010} only relies on triangle orientations and not on the exact geometry of the point set, their result transfers to the pseudolinear setting.
This allows us to generalize it to convex drawings in the same way as the Empty Hexagon Theorem (\Cref{thm:6holes}) using \Cref{lem:minimalkgonpseudolinear}.

\begin{corollary}
\label{cor:genmonochromatic4hole}
Every sufficiently large convex drawing on vertices $V=A \ \dot\cup \ B$ has an empty $4$-triangulation induced only by vertices from $A$ or only by vertices from~$B$.
\end{corollary}

As the following construction (\Cref{fig:new modifiedtwisted}) shows, there exist simple drawings of~$K_n$ without any empty $4$-triangulation.
For the construction, we start with the twisted drawing~$\calT_n$ and reroute some edges such that the drawing is still crossing maximal, that is, every 4-tuple contains a crossing.
The resulting drawing $\calT'_n$ then does not contain any empty $4$-triangulation and thus no $4$-hole.

\begin{proposition}
\label{prop:nogen4hole}
For $n \geq 6$ the simple drawing $\calT_n'$ contains no empty $4$-tri\-an\-gu\-la\-tion.
\end{proposition}

\begin{proof}
We start by giving the exact crossing edge pairs in $\calT_n'$ and thus describing the drawing up to isomorphism.
The vertices $1, 3, 4, \ldots ,n$ form a twisted drawing $\calT_{n-1}$ and hence every $4$-tuple from $[n] \backslash \{2\}$ contains a crossing, giving $\binom{n-1}{4}$ crossings.
More specifically, the edges $\{i,\ell\}$ and $\{j,k\}$ cross for $i,j,k,\ell \in [n] \setminus \{2\}$ with $i<j<k<\ell$.

It remains to describe the crossings in $4$-tuples which do contain vertex~$2$.
The edge $\{2,1\}$ crosses the edges $\{3,n\}$, $\{4,i\}$ for $i = 5, \ldots n$, and $\{3,4\}$, which are $n-2$ crossings.
The edge $\{2,3\}$ has no crossings and the edge $\{2,4\}$ crosses only the edge $\{3,n\}$.
For $j = 5, \ldots, n-1$ the edge $\{2,j\}$ crosses the two edges $\{3,n\}$ and $\{1,3\}$, the $n-j$ edges $\{1,j+1\},\ldots,\{1,n\}$, and the edges $\{i,k\}$ for $2<i<k<j$, of which there are $\binom{j-3}{2}$.
Finally, the edge $\{2,n\}$ crosses the $\binom{n-4}{2}$ edges $\{i,j\}$ for $3 < i < j < n$.

\begin{figure}[htb]
\centering
\includegraphics{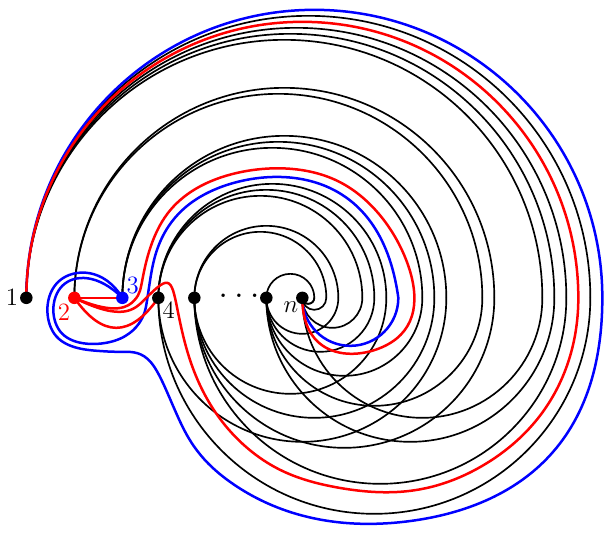}
\caption{The drawing $\calT_n'$ without empty $4$-triangulations for $n \geq 6$.}
\label{fig:new modifiedtwisted}
\end{figure}

In total there are 
\begin{align*}
&\binom{n-1}{4} + (n-2) + 1 + \sum_{j=5}^{n-1} \left( 2+(n-j)+\binom{j-3}{2} \right) + \binom{n-4}{2} \\
= \ &\binom{n-1}{4} + 3n-11 + \binom{n-4}{2} + \sum_{j=2}^{n-4} \binom{j}{2} + \binom{n-4}{2} \\
= \ &\binom{n-1}{4} + 2n-7 + \binom{n-3}{2} + \binom{n-3}{3} + \binom{n-4}{2} \\
= \ &\binom{n-1}{4} + 1+(n-4) + \binom{n-3}{2} + \binom{n-2}{3 } \\
= \ &\binom{n-1}{4} + \binom{n-1}{3} = \binom{n}{4}
\end{align*}
crossings because of the well-known identities $\sum_{j = r}^{n} \binom{j}{r} = \binom{n+1}{r+1}$ and $\sum_{k = 0}^{m} \binom{n+k}{k} = \binom{n+m+1}{m}$.
Hence $\calT_n'$ is crossing maximal.

Because of this crossing maximality every empty $4$-triangulation is a $4$-hole since the drawing of the four induced vertices has a crossing.
In the twisted subdrawing $\calT_{n-1}$ induced by $1,3,\ldots,n$ the empty triangles are $\{1,3,i\}$ for $i = 4, \ldots, n$ and $\{i, n-1,n\}$ for $i = 1, 3 \ldots n-2$ and the only 4-hole is $\{1,3,n-1,n\}$, which is not a 4-hole in $\calT'_n$ because we placed the vertex $2$ into the triangle $\{3,n-1,n\}$.
Hence if there is a $4$-hole, it consists of the vertex $2$ and the three vertices of an empty triangle of the induced subdrawing $\calT_{n-1}$.
However, since all empty triangles from the induced subdrawing $\calT_{n-1}$ have either both $1$ and $3$ or both $n-1$ and $n$ as vertices, at least one of the two triangles $\{1,2,3\}$ or $\{2,n-1,n\}$ must be empty.
This is not the case in the constructed drawing; the triangle $\{1,2,3\}$ has $4$ on one side and all other vertices on the other side and the triangle $\{2,n-1,n\}$ has $3$ on one side and all other vertices on the other side.
This completes the proof.
\end{proof}

In contrast to this construction, if instead of empty $4$-triangulations we only ask for empty $4$-cycles, then we can actually guarantee their existence in all simple drawings of $K_n$.
This resolves one case of a recent conjecture by Bergold, Felsner, M.\ Reddy, Orthaber, and Scheucher~\cite{bfmos-2025-phccd}.
They showed that every convex drawing contains an empty $k$-cycle for all $3 \leq k \leq n$ and conjectured that this also holds for simple drawings.

\begin{conjecture}[\cite{bfmos-2025-phccd}]
\label{conj:emptykcycles}
Every simple drawing of $K_n$ contains an empty $k$-cycle for each $3 \leq k \leq n$.
\end{conjecture}

While the case $k=3$ follows by Harborth's result~\cite{Harborth98_emptytriangles}, the case $k=n$ coincides with Rafla's conjecture concerning the existence of plane Hamiltonian cycles in all simple drawings of $K_n$~\cite{Rafla1988}.
For the proof of the case $k=4$ we use results on plane subdrawings by Garc\'ia, Pilz, and Tejel~\cite{gpt-2021-psctd}. 

\begin{theorem}
\label{thm:empty4cycle}
Let $\calD$ be a simple drawing of $K_n$ with $n \geq 4$ and let $v$ be a vertex of~$\calD$.
Then $\calD$ contains an empty $4$-cycle passing through~$v$.
\end{theorem}

\begin{proof}
For a fixed vertex~$v$, we consider the spanning star $S_v$ centered at~$v$.
By \cite[Corollary~3.4]{gpt-2021-psctd}, there is a plane subdrawing $\calD'$ of $\calD$ that consists of the star~$S_v$ and some spanning tree $T$ on the other $n-1$ vertices.
Note that $\calD'$ has exactly $2n-3$ edges and $n-1$ faces.
Every face $F$ of~$\calD'$ contains $v$ on its boundary because the tree $T$ is cycle-free and since $\calD'$ is 2-connected \cite[Theorem~3.1]{gpt-2021-psctd}, 
$F$ is bounded by exactly two edges of~$S_v$.

If there is a face of $\calD'$ with exactly $4$ boundary edges or if there are two adjacent triangular faces, we obtain an empty $4$-cycle passing through~$v$ and the statement follows.
Otherwise we count the number of edges $|E|$ in~$\calD'$:
At most half of the $n-1$ faces are triangles so that none of them are adjacent.
All other faces have at least $5$ boundary edges.
Since every edge is incident to exactly two faces, we have $|E| \geq \frac{1}{2} \left( 5(n-1) - 2\left\lfloor\frac{n-1}{2}\right\rfloor \right) \geq 2n-2$. 
This is a contradiction to the fact that $\calD'$ contains exactly $2n-3$ edges.
\end{proof}

The above theorem implies a linear lower bound on the number of empty $4$-cycles.
This is similar to the minimum number of empty triangles which is asymptotically linear as well~\cite{AichholzerHPRSV2015}.

\begin{corollary}
Every simple drawing of $K_n$ with $n \geq 4$ contains at least $\lceil \frac{n}{4} \rceil$ empty $4$-cycles.
\end{corollary}

While the twisted drawing $\calT_n$ is conjectured to minimize the number of empty triangles, it contains $\Theta(n^3)$ empty $4$-cycles, since the cycles $(i,j,l,k)$ for $i<j<k<l$ separate the elements between $j$ and $k$ from the rest, whereas all other $4$-cycles are crossing themselves.
This is certainly not minimal as, in the following, we construct drawings with $\Theta(n^2)$ empty $4$-cycles; see \Cref{fig:few4cycles}.

\begin{proposition}
\label{prop:empty4-cycle}
There is a simple drawing of $K_n$ that admits $\frac{1}{8}n^2+O(n)$ empty $4$-cycles.
\end{proposition}

Note that this is strictly less than the lower bound of $\frac{5}{2}n^2 - \Theta(n)$ for the number of empty $4$-cycles in geometric drawings shown in~\cite{afghhhuv-2014-4hps}.
Moreover, in the geometric setting, the number of empty $k$-cycles with $k \geq 4$ is actually conjectured to be super-quadratic~\cite{AICHHOLZER2015528}.

\begin{proof}
We start with the drawing $\calD_5 = \calC_5$ with vertices $1, \ldots, 5$ labeled counter-clockwise.
We then recursively construct the drawing $\calD_{n+1}$ from $\calD_n$ as follows:
We add a new vertex $n+1 \geq 6$ close to the vertex $n$ in a chosen cell $c_{n}$ next to the edge $e_{n} := \{ n, i_{n} \}$ for some choice of~$i_{n}$.
We connect $n+1$ and $n$ with an uncrossed edge.
Then we add the edges $\{ n+1, j \}$ for $j < n$ by making them cross all edges between $e_n$ and $\{ n, j \}$ incident to $n$ close to $n$ and then follow the edge $\{ n, j \}$ from $n$ to~$j$.
In particular, the edge $\{ n+1, i_n \}$ crosses all edges incident to $n$ except $e_n$ and $\{n+1,n\}$ before following $e_n$ to~$i_n$.
As shown by Harborth and Mengersen~\cite{HarborthMengersen1992} the resulting drawing $\calD_{n+1}$ is crossing maximal for all choices of $i_n$ and cells $c_n$ next to it.
Also note that by construction $n+1$ and $n$ have the same rotation (ignoring the edge between them).

For our construction, we assume that $n$ is odd and we perform two steps at once, hence producing drawings $\calD_{n+1}$ and $\calD_{n+2}$.
In the first step from $n$ to $n+1$, we choose $i_n=n-2$ and the cell $c_n$ not to be incident to $\{n,n-1\}$.
This is well-defined for the base case $n=5$ and also for all larger odd~$n$, as we make sure in the following that $n$ and $n+1$ are consecutive in the rotation of~$n+2$.
In the second step from $n+1$ to $n+2$, we choose $i_{n+1}=n$ and $c_{n+1}$ to be the cell not incident to $\{n+1,n-2\}$.
This cell is well-defined as we added the previous vertex $n+1$ in the cell incident to $e_{n}=\{n,n-2\}$.
We start with some general observations:
\begin{itemize}
\item Since the drawings are crossing maximal, every $4$-tuple of vertices can produce at most one empty $4$-cycle.
\item The vertices $n$, $n+1$, and $n+2$ have the exact same rotation (ignoring the edges between them), that is, removing a non-trivial subset of them from $\calD_{n+2}$ results in a drawing isomorphic to $\calD_{n}$ or $\calD_{n+1}$.
\item Every empty $4$-cycle in $\calD_{n}$ involving vertex $n$ that is still empty in $\calD_{n+2}$ produces two other empty $4$-cycles in $\calD_{n+2}$ involving vertex $n+1$ and $n+2$ respectively. These are the only $4$-cycles involving exactly one of $n$, $n+1$, and $n+2$.
\item An empty 4-cycle in $\calD_n$ still exists in $\calD_{n+2}$ if and only if it does not contain the cell $c_{n+1}$. In particular, only empty 4-cycles of $\calD_n$ incident to $n$ are destroyed by $n+1$ and $n+2$ as $c_{n+1}$ is incident to $n$ and any other 4-cycle containing them would contain $n$ as well.
\end{itemize}

\begin{figure}[htb]
\centering
\includegraphics[page=1]{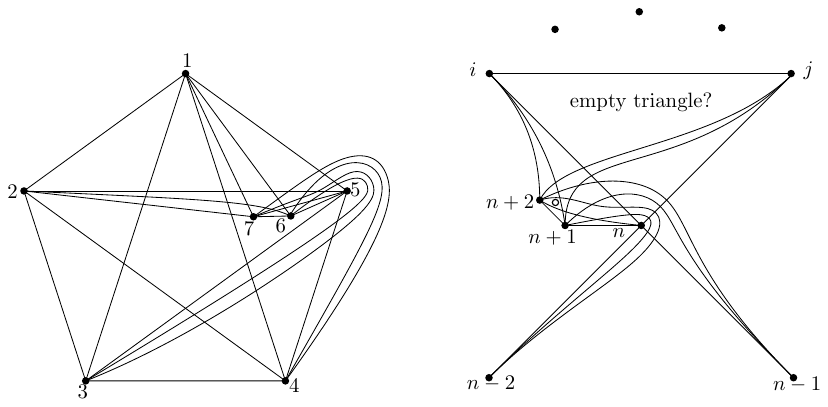}
\caption{Constructing the drawing $\calD_{n+2}$ of $K_{n+2}$ with few empty $4$-cycles for $n$ odd. The small circle indicates the cell where the additional vertices $n+3$ and $n+4$ will be put in the next step.}
\label{fig:few4cycles}
\end{figure}

We are left to characterize the empty $4$-cycles involving at least $2$ of $n$, $n+1$, and $n+2$:
The cycle $(i,n,n+1,n+2)$ is empty for all $i \leq n-1$; it is actually a $4$-hole.
In particular, the empty side contains $\{n,n+2\}$ completely, so this empty $4$-cycle will be destroyed in the next step when we introduce vertices just next to that edge.
See~\Cref{fig:destroyed4cycles} for an illustration.

Finally the $4$-cycles $(i,j,x,y), i<j<n\leq x<y$ are empty if and only if $\{i,j,n\}$ is an empty triangle in $\calD_{n+2}$.
See \Cref{fig:few4cycles} for an illustration.
Notably, empty triangles $\{i,j,n\}$ from $\calD_n$ that contain $n+1$ or $n+2$ also do not leave empty triangles nor empty 4-cycles with $n+1$ and/or $n+2$ since those contain $n$.

\begin{figure}[htb]
\centering
\includegraphics[page=2]{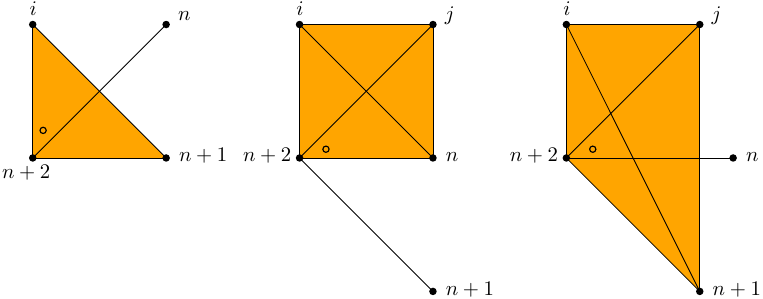}
\caption{The empty triangles and 4-cycles incident to $n+2$ that will not be empty anymore, once $n+3$ is added to the drawing. The empty side of these cycles is orange, while $n+3$ will be put at the location of the small circle.}
\label{fig:destroyed4cycles}
\end{figure}

It is therefore important to also consider which empty triangles are incident to $n+2$ after one step.
The ones of the form $\{i,j,n+2\}$, $i<j<n$ exist if and only if triangle $\{i,j,n\}$ is still empty in $\calD_{n+2}$ because of the second observation and the argumentation in the last paragraph.
Since $\{i,n,n+1,n+2\}$ is a $4$-hole for all $ i<n$, all other triangles are empty as well.

However, the empty $4$-cycles of the form $(i,j,x,n+2)$ as well as triangles of the form $\{i,x,n+2\}$ are again destroyed in the next step.
See~\Cref{fig:destroyed4cycles} for an illustration. 
If $x=n$, this is true because for $n+1$ not to be in the empty side, since $\{n+1,n+2\}$ is uncrossed, the empty side has to be on the other side of $\{n,n+2\}=e_{n+2}$. 
If $x=n+1$, we know $\{n,n+2\}=e_{n+2}$ is crossed only by edges incident to $n+1$ such as $\{i,n+1\}$ and $\{j,n+1\}$ so the empty side has to contain the first segment of it completely.

Thus the only empty 4-cycles introduced in this step, which will stay empty through the next step are of the form $(i,j,n,n+1)$ for each empty triangle $\{i,j,n\}$, whereas the empty triangles $\{i,j,n\}$ that are still empty after that step give rise to empty triangles $\{i,j,n+2\}$ and there is a single additional empty triangle $\{n,n+1,n+2\}$.
From \Cref{fig:few4cycles} it is easy to convince yourself, that after the first step, the only empty triangles that vertex $5$ is still incident to are $\{1,2,5\}$ and $\{3,4,5\}$.
Thus the empty triangles incident to the last vertex $n$ are going to be all triangles of the kind $\{2k-1,2k,n\}$ for some $k<\frac{n}{2}$ and the empty $4$-cycles that will stay are of the form $(2i-1,2i,2j-1,2j)$ for $i<j\leq\frac{n}{2}$.
These are only $\binom{\lfloor\frac{n}{2}\rfloor}{2} = \frac{1}{8}n^2 + O(n)$ empty $4$-cycles and the linear number of additional empty $4$-cycles incident to $n$ of the forms $(i,n-2,n-1,n), i<n-2$ and $(2k-1,2k,x,n), k<\frac{n-2}{2}, x\in\{n-2,n-1\}$ do not change these asymptotics.
\end{proof}

\section{Conclusion}\label{sec:conclusion}

We have shown that every convex drawing of $K_n$ with $n \geq 5$ contains a quadratic number of $4$-holes and that sufficiently large convex drawings contain $5$- and $6$-holes, while $7$-holes do not exist in general.
For $k \in \{ 5, 6 \}$ given, it remains an interesting open question to determine the smallest integer $n(k)$ such that every convex drawing of $K_n$ with $n \geq n(k)$ contains a $k$-hole.

In the geometric setting, Harborth~\cite{Harborth1978_5holes} showed that $10$ points are sufficient to contain a $5$-hole and since recently it is known that $30$ points in general position always contain a $6$-hole~\cite{HeuleScheucher2024,Overmars03}.
Note that for larger point sets we can find a $6$-hole in the $30$ leftmost points, which is still a $6$-hole in the whole point set.
By this argument, containing a $k$-hole is a monotone property for point sets.
In contrast to that, we used the SAT framework from \cite{bs-2025-isdknsat-arxiv} to find convex drawings for $n \leq 10$ and $n=12$ without $5$-holes, while proving that every convex drawing for $n =11$ and $13 \leq n \leq 16$ contains a $5$-hole.
This shows that containing a $k$-hole is in general not a monotone property for convex drawings.
Based on the computational data, however, we conjecture that every convex drawing on at least $13$ vertices contains a $5$-hole.

It would further be interesting to obtain better bounds on the size of a largest $k$-gon and on the size of a largest $f$-convex subdrawing in a convex drawing of~$K_n$.
The currently best estimate for a $k$-gon is by Suk and Zeng~\cite{sz-2024-upcstg}, which yields $\Omega( (\log n)^{1/2-o(1)} )$, and combining this with \Cref{lem:minimalkgonpseudolinear} yields an $f$-convex drawing of the same size.

Furthermore, it would be intriguing to define some kind of twisted hole of size $k$ based on the twisted drawing~$\calT_k$.
Then, in the flavor of the result that every simple drawing of $K_n$ contains either a $k$-gon or twisted subdrawing of a certain size~\cite{PachSolymosiToth2003,sz-2024-upcstg}, one could try to show the existence of a $6$-hole or twisted hole of size $6$ in every simple drawing of $K_n$.
Clearly a twisted hole of size $k$ would have to be defined via a plane $k$-cycle.
However, while in a $k$-gon the plane $k$-cycle is unique, the twisted drawing $\calT_k$ contains exponentially many plane $k$-cycles~\cite{o-2022-cfhcsdcg} and it is unclear which of them would be a somehow canonical choice.

Finally, while our construction for few empty $4$-cycles shows a clear difference to the geometric setting, it remains an open question whether a sub-quadratic or even just a linear number of empty $4$-cycles is possible in a simple drawing of $K_n$.

\bibliography{references.bib}

\end{document}